\documentclass[a4paper, 10pt, english]{article}

\pdfoutput=1

\usepackage[T1]{fontenc}        
\usepackage{lmodern}            
\usepackage[latin1]{inputenc}   
\usepackage[english]{babel}              

\usepackage{amssymb}
\usepackage{amsmath}
\usepackage{amsfonts}
\usepackage{amsthm}

\newtheorem{lemma}{Lemma}
\newtheorem{remark}{Remark}
\newtheorem{example}{Example}

\usepackage{graphicx}           
\usepackage{hyperref}

\usepackage{dapabrv}
\usepackage{dapmath}

\newcommand{\figwidth}{.85\textwidth}
\newcommand{\figComment}{For a more easy read plot, in color, the reader is referred to the digital version of this paper.}

\date{\today}

\begin{document}


  \title{Model Reduction using a Frequency-Limited $\Htwo$-Cost}

\author{Daniel~Petersson$^ a$\thanks{ Corresponding
    author. Email: petersson@isy.liu.se \vspace{6pt}} ~ and
  Johan~L\"ofberg$^ a$\\
\vspace{6pt} \\
$^ a$ Division of Automatic Control, \\
Department of Electrical Engineering,\\
Link\"opings universitet, SE-581 83 Sweden }

\maketitle
  \begin{abstract}
    We propose a method for model reduction on a given frequency range, without the use of input and output filter weights. The method uses a nonlinear optimization approach to minimize a frequency limited $\Htwo$ like cost function.
    An important contribution in the paper is the derivation of the gradient of the proposed cost function. The fact that we have a closed form expression for the gradient and that considerations have been taken to make the gradient computationally efficient to compute enables us to efficiently use off-the-shelf optimization software to solve the optimization problem.
  \end{abstract}


\section{Introduction}

Given a linear time-invariant (\lti) dynamical model,
\begin{align*}
  \dot x(t) &= \m Ax(t) + \m Bu(t), \\
  y(t) &= \m Cx(t) + \m Du(t)
\end{align*}
where $\m A\in\reals^{n\times n}$, $\m B\in\reals^{n\times m}$, $\m C \in \reals^{p\times n}$ and $\m D\in \reals^{p\times m}$, the model reduction problem is to find a reduced order model
\begin{align*}
  \dot x_r(t) &= \m A_rx_r(t) + \m B_ru(t), \\
  y_r(t) &= \m C_rx_r(t) + \m D_ru(t),
\end{align*}
with $\m A_r\in\reals^{n_r\times n_r}$, $\m B_r\in\reals^{n_r\times m}$, $\m C_r \in \reals^{p\times n_r}$ and $\m D_r\in \reals^{p\times m}$ with $n_r < n$, where this reduced order model describes the original model well in some metric. In this paper we are interested in a reduced order model that describes the model well on a given frequency range. This is motivated by situations where the given model is valid only for a certain frequency range, for example as in \cite{VargaHP:2012} where models coming from aerodynamical and structural mechanics computations describing a flexible structure are only valid up to a certain frequency. 

For a review of model reduction approaches, both ordinary and frequency-weighted, see \eg \cite{GugercinA:2004} and \cite{GhafoorS:2008}. Some of the most commonly used frequency-weighted methods, according to \cite{GugercinA:2004}, are \cite{WangSL:1999}, \cite{Enns:1984} and \cite{LinC:1992}, which all use different balanced truncation approaches. In many of the frequency-weighted methods one has to specify input and output filter weights. In \cite{GawronskiJ:1990} they introduce a method which does not need these weighting functions, by introducing frequency-limited Gramians. This method can be interpreted as using ideal low-, band- or high-pass filters as weights. However, this method has the drawback of not always producing stable models. One approach to remedy this has been presented in \cite{GugercinA:2004}, where they introduce a modification of the method in \cite{GawronskiJ:1990}, and additionally derive an $\Hinf$ bound for the error. \cite{SahlanGS:2012} also presents a modification of the method from \cite{GawronskiJ:1990}, however this method is only applicable to \siso models. Note that these methods, whilst having an upper bound on the $\Hinf$ error, do not minimize an explicit measure.

One of the main contributions in this paper is a method which uses optimization, and not truncation, to find an $\Htwo$-optimal reduced order model, and, where this optimization is only performed over a limited frequency interval. Another important contribution is the derivation of the gradient for the cost function. The fact that we have a closed form expression for the gradient and that considerations have been taken to make the gradient computationally efficient to compute enables us to efficiently use off-the-shelf optimization software to solve the optimization problem.


\section{Frequency Limited Gramians}
\label{sec:freq-limit-gram}

The method proposed in this paper uses the idea presented in \cite{GawronskiJ:1990}, in which they introduce frequency-limited Gramians.  Before introducing the frequency-limited Gramians we define the standard Gramians, see \cite{ZhouDG:1996}, in time and frequency domain, for reference.
For a system $G$, that is stable and described by
\begin{align}
  \label{eq:7}
  G:\left\{
    \begin{array}{rcl}
      \dot x(t) & = & \m A x(t) + \m B u(t) \\
      y(t) & = & \m C x(t) + \m D u(t)
    \end{array}
  \right.
\end{align}
denoted as $G : \systemdesc{\m A}{\m B}{\m C}{\m D}$, the observability and controllability Gramians are defined as
\begin{subequations}\label{eq:GramDef}
  \begin{align}
    \m P & = \int_0^\infty \e{\m A\tau}\m B\m B^\+ \e{\m A^\+\tau}\der\tau = \frac{1}{2\pi}\int_{-\infty}^\infty \m H(\nu)\m B\m B^\+\m H^\herm(\nu) \der\nu, \\
    \m Q & = \int_0^\infty \e{\m A^\+\tau}\m C^\+\m C \e{\m A\tau}\der\tau = \frac{1}{2\pi}\int_{-\infty}^\infty \m H^\herm(\nu)\m C^\+\m C\m H(\nu) \der\nu,
  \end{align}
\end{subequations}
where $\m H(\omega) = (\ii\omega\ident - \m A)^\inv$ and $\m H^\herm(\omega)$ denotes the conjugate transpose of $\m H(\omega)$. The controllability and observability Gramians satisfy, respectively, the Lyapunov equations
\begin{subequations}\label{eq:lyap}
\begin{align}
  \m A\m P + \m P\m A^\+ + \m B\m B^\+ & = \m 0, \\
  \m A^\+\m Q + \m Q\m A + \m C^\+\m C & = \m 0.
\end{align}
\end{subequations}

Now we narrow the frequency band, from $(-\infty,\infty)$ to $(-\omega,\omega)$ where $\omega < \infty$. We define the frequency-limited Gramians, see \cite{GawronskiJ:1990}, as
\begin{subequations}\label{eq:LimGramDef}
  \begin{align}
    \m P_{\omega} & = \frac{1}{2\pi}\int_{-\omega}^\omega \m H(\nu)\m B\m B^\+\m H^\herm(\nu) \der\nu, \\
    \m Q_{\omega} & =\frac{1}{2\pi}\int_{-\omega}^\omega \m H^\herm(\nu)\m C^\+\m C\m H(\nu) \der\nu.
  \end{align}
\end{subequations}
These Gramians can be shown to satisfy the following Lyapunov equations, see \cite{GawronskiJ:1990},
\begin{subequations}
\begin{align}
  \m A\m P_{\omega} + \m P_{\omega}\m A^\+ + \m S_{\omega}\m B\m B^\+ + \m B\m B^\+\m S^\herm_{\omega} & = \m 0, \\
  \m A^\+\m Q_{\omega} + \m Q_{\omega}\m A + \m S^\herm_{\omega}\m C^\+\m C + \m C^\+\m C\m S_{\omega} & = \m 0,
\end{align}
\end{subequations}
with 
\begin{equation}
  \m S_{\omega} = \frac{\ii}{2\pi}\ln\parens{(\m A + \ii\omega\ident)(\m A - \ii\omega\ident)^\inv}
\end{equation}

In \cite{GawronskiJ:1990} they continue by creating a balanced system and performing a balanced truncation using the newly defined frequency-limited Gramians. A drawback with this method is that since the terms $\m S_{\omega}\m B\m B^\+ + \m B\m B^\+\m S^\herm_{\omega}$ and $\m S^\herm_{\omega}\m C^\+\m C + \m C^\+\m C\m S_{\omega}$ are not guaranteed to be positive definite, stability of the reduced order model cannot be guaranteed. There exist a modification to this in \cite{GugercinA:2004} where they propose a remedy to this.

\begin{remark}
  By using addition/subtraction of two or more different
  frequency-limited Gramians it is possible to focus on one or more
  arbitrary frequency ranges, \eg, you can construct the
  frequency-limited controllability Gramian, $\m P_\Omega$, for the
  interval
  $\omega\in\Omega=[\omega_1,\omega_2]\cup[\omega_3,\omega_4]$ as
  \begin{equation}
    \label{eq:15}
    \m A\m P_{\Omega} + \m P_{\Omega}\m A^\+ + \m S_{\Omega}\m B\m B^\+ + \m B\m B^\+\m S^\herm_{\Omega} = \m 0,
  \end{equation}
  with $\m S_{\Omega} = \m S_{\omega_2} - \m S_{\omega_1} + \m
  S_{\omega_4} - \m S_{\omega_3}$.
\end{remark}

\section{Frequency Limited Model Reduction using Optimization}
\label{sec:freq-limit-htwo}

The $\Htwo$-norm of $G$, in \eqref{eq:7}, can be expressed as
\begin{subequations}
\begin{align}
  \label{eq:1}
  \norm{G}^2_{\Htwo} = & \trace \int_0^\infty \m C\e{\m A\tau}\m B\m B^\+ \e{\m A^\+\tau}\m C^\+ \der\tau \\
= & \frac{1}{2\pi}\trace \int_{-\infty}^\infty G(\ii\nu)G^\herm(\ii\nu) \der\nu \\
  = & \frac{1}{2\pi}\trace\int_{-\infty}^\infty \m C\m H(\nu)\m B\m B^\+\m H^\herm(\nu)\m C^\+ \der\nu = \trace \m C \m P \m C^\+ \\
  = & \trace \int_0^\infty \m B^\+ \e{\m A^\+\tau}\m C^\+\m C \e{\m A\tau}\m B \der\tau \\
  = &\frac{1}{2\pi}\trace\int_{-\infty}^\infty \m B^\+\m H^\herm(\nu)\m C^\+\m C\m H(\nu) \m B\der\nu = \trace \m B^\+ \m Q \m B.
\end{align}
\end{subequations}
In this paper we introduce a new frequency-limited $\Htwo$-like norm that uses the frequency-limited Gramians presented in the previous section, and we denote the new measure by $\norm{G}_{\Htwo,\omega}$, with
\begin{align}
  \label{eq:8}
  \norm{G}_{\Htwo,\omega}^2 = & \frac{1}{2\pi}\trace \int_{-\omega}^\omega G(\ii\nu)G^\herm(\ii\nu) \der\nu \\
= & \frac{1}{2\pi}\trace\int_{-\omega}^\omega \parens{\m C\m H(\nu)\m B + \m D}\parens{\m B^\+\m H^\herm(\nu)\m C^\+ + \m D^\+} \der\nu \\
  = & \trace \m C \m P_\omega \m C^\+ + 2\trace\brackets{\parens{\m C \m S_\omega\m B + \m D\frac{\omega}{2\pi}}\m D^\+}.
\end{align}
One thing that differs from the ordinary $\Htwo$-norm is that, if we do not include an infinite interval in $\Omega$, \ie, include $\omega = \infty$ as the end frequency, then the system does not need to be strictly proper. This means that we can, in this case, have $\m D \neq \m 0$.

The method proposed in this paper is a model reduction method that, given a model $G$, finds a reduced order model, $\hat G$, that is a good approximation on a given frequency interval, \eg $[0,\omega]$. The objective is to minimize the error between the given model and the sought reduced order model in a frequency-limited $\Htwo$-norm, using the frequency-limited Gramians. We formulate the optimization problem
\begin{align}
  \label{eq:3}
  \minimize{\hat G} \norm{G-\hat G}^2_{\Htwo,\omega} = \minimize{\hat G} \norm{E}^2_{\Htwo,\omega},
\end{align}
where
\begin{align}
  \label{eq:12}
  \norm{E}^2_{\Htwo,\omega} = \frac{1}{2\pi}\trace\int_{-\omega}^\omega E(\ii\nu) E^\herm (\ii\nu) \der\nu.
\end{align}
Assume that the system $E$ is stable and described by
\begin{align}
  E : \systemdesc{\m A_E}{\m B_E}{\m C_E}{\m D_E}.
\end{align}
Given $G$ and $\hat G$, represented as
\begin{equation}
  \label{eq:5}
   G : \systemdesc{\m A}{\m B}{\m C}{\m D}, 
   \hat G : \systemdesc{\mh A}{\mh B}{\mh C}{\mh D},
\end{equation}
the error system can be realized, in state space form, as
\begin{align}\label{eq:partitioningE}
  E & : \systemdesc{\m A_{E}}{\m B_{E}}{\m C_{E}}{\m D_{E}} = \systemdesc%
  {
    \begin{pmatrix}
      \m A & \m 0 \\
      \m 0 & \mh A
    \end{pmatrix}
  }
  {
    \begin{pmatrix}
      \m B \\ \mh B
    \end{pmatrix}
  }
  {
    \begin{pmatrix}
      \m C & -\mh C
    \end{pmatrix}
  }
  {\m D - \mh D}.
\end{align}
This realization of the error system will later prove beneficial when rewriting the optimization problem. Throughout the paper we will assume that the given model is stable.

The cost function for the optimization problem \eqref{eq:3} can be written as
\begin{subequations}
  \label{eq:2}
  \begin{align}
    \norm{E}^2_{\Htwo,\omega} = & \trace \m C_E \m P_{E,\omega} \m C_E^\+ + 2\trace\brackets{\parens{\m C_E \m S_{E,\omega}\m B_E + \m D_E\frac{\omega}{2\pi}}\m D_E^\+} \\
    = &  \trace \m B_E^\+ \m Q_{E,\omega} \m B_E + 2\trace\brackets{\parens{\m C_E \m S_{E,\omega}\m B_E + \m D_E\frac{\omega}{2\pi}}\m D_E^\+}.
  \end{align}
\end{subequations}
where
\begin{subequations}\label{eq:lyapLim}
\begin{align}
  \m A_E\m P_{E,\omega} + \m P_{E,\omega}\m A_E^\+ + \m S_{E,\omega}\m B_E\m B_E^\+ + \m B_E\m B_E^\+\m S^\herm_{E,\omega} & = \m 0, \\
  \m A_E^\+\m Q_{E,\omega} + \m Q_{E,\omega}\m A_E + \m S^\herm_{E,\omega}\m C_E^\+\m C_E + \m C_E^\+\m C_E\m S_{E,\omega} & = \m 0,
\end{align}
\end{subequations}
with 
\begin{equation}
  \m S_{E,\omega} = \frac{\ii}{2\pi}\ln\parens{(\m A_E + \ii\omega\ident)(\m A_E - \ii\omega\ident)^\inv}.
\end{equation}
In this paper we have also derived a simpler expression for the matrix $\m S_{E,\omega}$, compared to what is presented in \cite{GawronskiJ:1990}.
\begin{lemma}\label{lem:realPart}
  For a matrix $\m A$ that is Hurwitz we have that
  \begin{equation}
    \label{eq:16}
    \m S_{\omega} = \frac{\ii}{2\pi}\ln\parens{(\m A + \ii\omega\ident)(\m A - \ii\omega\ident)^\inv} = 
    \real\brackets{\frac{\ii}{\pi}\ln\parens{-\m A-\ii\omega\ident}}
  \end{equation}
\end{lemma}
\begin{proof}
  See \ref{sec:proof-lemma-}
\end{proof}

Now we want to rewrite the cost function \eqref{eq:2} to a more computationally tractable form. This is done by using the realization given in \eqref{eq:partitioningE} and by partitioning the Gramians $\m P_{E,\omega}$ and $\m Q_{E,\omega}$ as 
\begin{equation}\label{eq:grampart}
  \m P_{E,\omega} =
  \begin{pmatrix}
    \m P_\omega & \m X_\omega \\
    \m X_\omega^\+ & \mh P_\omega
  \end{pmatrix}, \quad
  \m Q_{E,\omega} =
  \begin{pmatrix}
    \m Q_\omega & \m Y_\omega \\
    \m Y_\omega^\+ & \mh Q_\omega
  \end{pmatrix},
\end{equation}
and $\m S_{E,\omega}$ as
\begin{align}
  \label{eq:10}
  \m S_{E,\omega} = 
  \begin{pmatrix}
      \m S_\omega & \m 0 \\
      \m 0 & \mh S_\omega
    \end{pmatrix}.
\end{align}
$\m P_\omega, \m Q_\omega,\mh P_\omega, \mh Q_\omega,\m X_\omega$ and $\m Y_\omega$ satisfy, due to (\ref{eq:lyapLim}), the Sylvester and Lyapunov equations
\begin{subequations}\label{eq:lyapsyl}
\begin{align}
  \m A\m P_\omega + \m P_\omega\m A^\+ + \m S_\omega\m B\m B^\+ + \m B\m B^\+\m S^\herm_\omega & = \m 0, \label{eq:lyapsyl1}\\
  \m A\m X_\omega + \m X_\omega\mh A^\+ + \m S_\omega\m B\mh B^\+ + \m B\mh B^\+\mh S^\herm_\omega & = \m 0, \label{eq:lyapsyl2}\\
  \mh A\mh P_\omega + \mh P_\omega\mh A^\+ + \mh S_\omega\mh B\mh B^\+ + \mh B\mh B^\+\mh S^\herm_\omega & = \m 0, \label{eq:lyapsyl3}\\
  \m A^\+\m Q_\omega + \m Q_\omega\m A + \m S^\herm_\omega\m C^\+\m C + \m C^\+\m C\m S_\omega & = \m 0, \label{eq:lyapsyl4}\\
  \m A^\+\m Y_\omega + \m Y_\omega\mh A - \m S^\herm_\omega\m C^\+\mh C - \m C^\+\mh C\mh S_\omega & = \m 0, \label{eq:lyapsyl5}\\
  \mh A^\+\mh Q_\omega + \mh Q_\omega\mh A + \mh S^\herm_\omega\mh C^\+\mh C + \mh C^\+\mh C\mh S_\omega & = \m 0, \label{eq:lyapsyl6}
\end{align}
\end{subequations}
with
\begin{equation}
  \m S_{\omega} = \real\brackets{\frac{\ii}{2\pi}\ln\parens{-\m A -\ii\omega\ident}}, \quad \mh S_{\omega} =  \real\brackets{\frac{\ii}{2\pi}\ln\parens{-\mh A -\ii\omega\ident}}.
\end{equation}

Note that $\m P_\omega$ and $\m Q_\omega$ satisfy the Lyapunov equations for the frequency-limited controllability and observability Gramians for the given model, and $\mh P_\omega$ and $\mh Q_\omega$ satisfy the Lyapunov equations for the frequency-limited controllability and observability Gramians for the sought model.

With the partitioning of $\m P_{E,\omega}$ and $\m Q_{E,\omega}$ it is possible to rewrite \eqref{eq:2} in two alternative forms
\begin{subequations}\label{eq:costfcns}
  \begin{align}
    \norm{E}^2_{\Htwo,\omega} = & \trace \left( \m B^\+\m Q_\omega\m B + 2\m B^\+\m Y_\omega\mh B + \mh B^\+\mh Q_\omega\mh B\right) \nonumber \\
    & + 2\trace\brackets{\m C\m S_\omega\m B + \m D\frac{\omega}{2\pi} -  \parens{\mh C\mh S_\omega\mh B + \mh D\frac{\omega}{2\pi}}}\parens{\m D^\+ - \mh D^\+} \label{eq:costfcnB}, \\
    \norm{E}^2_{\Htwo,\omega} = & \trace \left( \m C\m P_\omega\m C^\+ - 2\m C\m X_\omega\mh C^\+ + \mh C\mh P_\omega\mh C^\+\right) \nonumber \\
    & + 2\trace\brackets{\m C\m S_\omega\m B + \m D\frac{\omega}{2\pi} -  \parens{\mh C\mh S_\omega\mh B + \mh D\frac{\omega}{2\pi}}}\parens{\m D^\+ - \mh D^\+}. \label{eq:costfcnC}
  \end{align}
\end{subequations}
\begin{remark}
  Note that neither the term $\m B^\+\m Q_\omega\m B$ nor the term $\m C\m P_\omega\m C^\+$, which are included in the cost function \eqref{eq:costfcns}, depend on the optimization variables, $\mh A,\mh B,\mh C$ and $\mh D$. Hence, these terms can be excluded from the optimization. These are the only terms including $\m P_\omega$ and $\m Q_\omega$ which are the most costly to compute.
\end{remark}

When optimizing the frequency-limited $\Htwo$-norm using the system matrices as optimization variables we have the freedom to choose which elements we optimize over, \ie, we can introduce structure in $\mh A,\mh B,\mh C$ and $\mh D$, as long as we can find an $\mh A$ that is Hurwitz. Let us introduce the matrices $\m S_{\mh A},\m S_{\mh B},\m S_{\mh C}$ and $\m S_{\mh D}$ which hold the structure of the sought matrices, \ie,
\begin{align}
  \label{eq:9}
  \brackets{\m S_{\mh A}}_{ij} = \left\{
    \begin{matrix}
      1, & \quad \text{if } \brackets{\mh A}_{ij} \text{ is a free variable;} \\
      0, & \quad \text{otherwise.} \\
    \end{matrix}
    \right.
\end{align}
We will see in the next section, that due to the element-wise differentiation, this structure will be inherited in the gradient.

The parametrization of the sought system using the full system matrices is of course redundant, which leads to a non-unique minimum of the cost function in the parameter space. This leads to a singular Hessian matrix. However, this is taken care of in most quasi-newton solvers to ensure that the minimum is reached in a numerically stable way.

\subsection{Gradient of the Cost Function}
\label{sec:grad-cost-funct}

An appealing feature of the proposed nonlinear optimization approach, using our proposed $\Htwo$-like measure to solve the problem, is that the equations (\ref{eq:costfcns}) are differentiable in the system matrices, $\mh A,\mh B,\mh C$ and $\mh D$. In addition, the closed form expression obtained when differentiating the cost function is expressed in the given data ($\m A,\m B,\m C$ and $\m D$), the optimization variables ($\mh A,\mh B,\mh C$ and $\mh D$) and solutions to the equations in \eqref{eq:lyapsyl}.

To show this we start by differentiating with respect to $\mh B,\mh C$ and $\mh D$. First we note that neither $\m Q_\omega, \m Y_\omega$ nor $\mh Q_\omega$ in equation \eqref{eq:costfcnB} depends on $\mh B$ which means that the equation is quadratic in $\mh B$. Analogous observations can be made with equation \eqref{eq:costfcnC} and the variable $\mh C$ and similarly with $\mh D$. Hence, the derivative of the cost function with respect $\mh B,\mh C$ and $\mh D$ becomes
\begin{subequations}
\begin{align}
  \frac{\partial\norm{E}^2_{\Htwo,\omega}}{\partial\mh B} & = 2\left( \mh Q_\omega\mh B + \m Y_\omega^\+\m B - \mh S_\omega^\+ \mh C^\+\parens{\m D - \mh D}\right)\odot \m S_{\mh B}, \\
  \frac{\partial\norm{E}^2_{\Htwo,\omega}}{\partial\mh C} & = 2\left( \mh C\mh P_\omega - \m C\m X_\omega - \parens{\m D - \mh D}\mh B^\+\mh S_\omega^\+\right) \odot \m S_{\mh C}, \\
  \frac{\partial\norm{E}^2_{\Htwo,\omega}}{\partial\mh D} & = -2\parens{\m C\m S_\omega\m B + \m D\frac{\omega}{2\pi} - \mh C\mh S_\omega\mh B- \mh D\frac{\omega}{2\pi} + \parens{\m D - \mh D}\frac{\omega}{2\pi}} \odot \m S_{\mh D},
\end{align}
\end{subequations}
where $\odot$ represents the Hadamard product of matrices, \ie, element-wise multiplication.

For the more complicated case of differentiating with respect to $\mh A$ we observe that $\mh Q_\omega$ and $\m Y_\omega$ do depend on $\mh A$, see the equations in \eqref{eq:lyapsyl}. The calculations of this part of the gradient are lengthy and can be found in \ref{sec:derGradA}.

The complete gradient becomes
\begin{subequations}
\begin{align}
  \frac{\partial\norm{E}^2_{\Htwo,\omega}}{\partial\mh A} = & 2\parens{\m Y_\omega^\+\m X + \mh Q_\omega\mh P}\odot \m S_{\mh A} - 2\m W \odot \m S_{\mh A}, \\
\frac{\partial\norm{E}^2_{\Htwo,\omega}}{\partial\mh B} = & 2\left( \mh Q_\omega\mh B + \m Y_\omega^\+\m B \mh S_\omega^\+ \mh C^\+\parens{\m D - \mh D}\right) \odot \m S_{\mh B}, \\
  \frac{\partial\norm{E}^2_{\Htwo,\omega}}{\partial\mh C} = & 2\left( \mh C\mh P_\omega - \m C\m X_\omega - \parens{\m D - \mh D}\mh B^\+\mh S_\omega^\+\right) \odot \m S_{\mh C}, \\
  \frac{\partial\norm{E}^2_{\Htwo,\omega}}{\partial\mh D} = & -2\parens{\m C\m S_\omega\m B + \m D\omega - \mh C\mh S_\omega\mh B- \mh D\omega + \parens{\m D - \mh D}\omega} \odot \m S_{\mh D},
\end{align}
\end{subequations}
where
\begin{subequations}
  \label{eq:11}
  \begin{align}
    \m W = & \real\left[\frac{\ii}{\pi} L\parens{-\mh A-\ii\omega\ident,\m V}\right]^\+, \\
    \m V = & \mh C^\+\mh C\mh P - \mh C^\+\m C\m X-\mh C^\+\parens{\m
      D - \mh D}\mh B^\+
  \end{align}
\end{subequations}
with the function $L(\cdot,\cdot)$ being the Frech\'et derivative of the matrix logarithm, see \cite{Higham:2008}.

\begin{remark}
  Remember that if $\omega = \infty$ is included as the end frequency
  we need to have $\m D_E = \m D - \mh D = \m 0$, \ie, $\mh D$ constant and with $\m D = \mh D$.
\end{remark}
\begin{remark}
  The proposed method can, analogous to what is done in \cite{Petersson:2010}, easily be extended to a method for identifying \lpv-models over a limited frequency domain.
\end{remark}
\begin{remark}
  By supplying the cost function and its gradient in computationally efficient forms this method can be used in any off-the-shelf Quasi-Newton solver.
\end{remark}
\begin{remark}
  By using a stable model, \eg, a model from a Hankel reduction, as an initial point in the optimization and using a line-search we can limit the search to stable models.
\end{remark}

\section{Numerical Examples}
\label{sec:numerical-examples}

In this section three examples are used to illustrate the applicability of the method and to compare it with other methods. In the examples we will use three different methods; Truncation of Hankel singular values (will be called Hankel), the method proposed in \cite{GawronskiJ:1990} (called Gawronski), the method proposed in \cite{GugercinA:2004} (called Mod. Gawronski) and the proposed method (called Prop. method). The proposed method and the original and modified Gawronski method will take the limited frequency range into account, but Hankel will not. We will also compare with the Hankel method where we use an input filter to help that method to focus on the frequency interval of interest. The Gawronski method is a representative method among frequency-weighted methods, see \cite{GugercinA:2004}, with the benefit of not having to design weighting functions, but with the drawback that it cannot guarantee that the resulting model is stable. 

The proposed method uses a cost function which is non-convex, which makes it important to use a good initial point. For the examples presented here we have used the model obtained by the Gawronski method as an initial point for the optimization in the proposed method. If the Gawronski method generates an unstable model we use the Hankel method instead.

To evaluate the different models against each other we will compare the error model, $G-\hat G$, for the given frequency interval using the limited frequency $\Htwo$-norm, $\norm{G-\hat G}_{\Htwo,\omega}$, and the relative limited frequency $\Htwo$-norm, $\frac{\norm{G-\hat G}_{\Htwo,\omega}} {\norm{G}_{\Htwo,\omega}}$. We will also compute the relative $\Hinf$-norm on the given frequency interval, denoted $\frac{\norm{G-\hat G}_{\Hinf,\omega}}{\norm{G}_{\Hinf,\omega}}$, and the eigenvalue with the largest real part.

\begin{example}[Small illustrative example]\label{ex:twoModeSys}
  This example addresses a small model with four states. The model is composed of two second order models in series, one with a resonance frequency at $\omega=1$ and the other at $\omega=3$. We will limit the frequency range to $\omega\in[0,1.7]$ to try to only capture the first model.
  \begin{align}
    \label{eq:4}
    G = G_1G_2 = \frac{1}{s^2+0.2s+1}\frac{9}{s^2+0.003s+9}.
  \end{align}
To help the Hankel method we create a low pass Butterworth filter of order 10 with a cut off frequency of 1.7, see Figure \ref{fig:twoModeSysFilter}.
The results from the different methods can be seen in Figures \ref{fig:twoModeSysFull} and \ref{fig:twoModeSysError} and Table \ref{tab:twoModeSys}. As can be seen in the result we are successful in finding a good model for the first model with both the proposed method and the Gawronski method. All the reduced order models are stable. The Hankel method and the modified Gawronski method captures the wrong resonance mode (from our perspective) and fails completely in the lower frequency region. The Hankel method on the filtered model finds a model with the same resonance frequency, but otherwise has a bad correspondence with the true model. The proposed method and the Gawronski method return models essentially indiscernible.
\begin{figure}
   \centering
  \includegraphics[width=\figwidth]{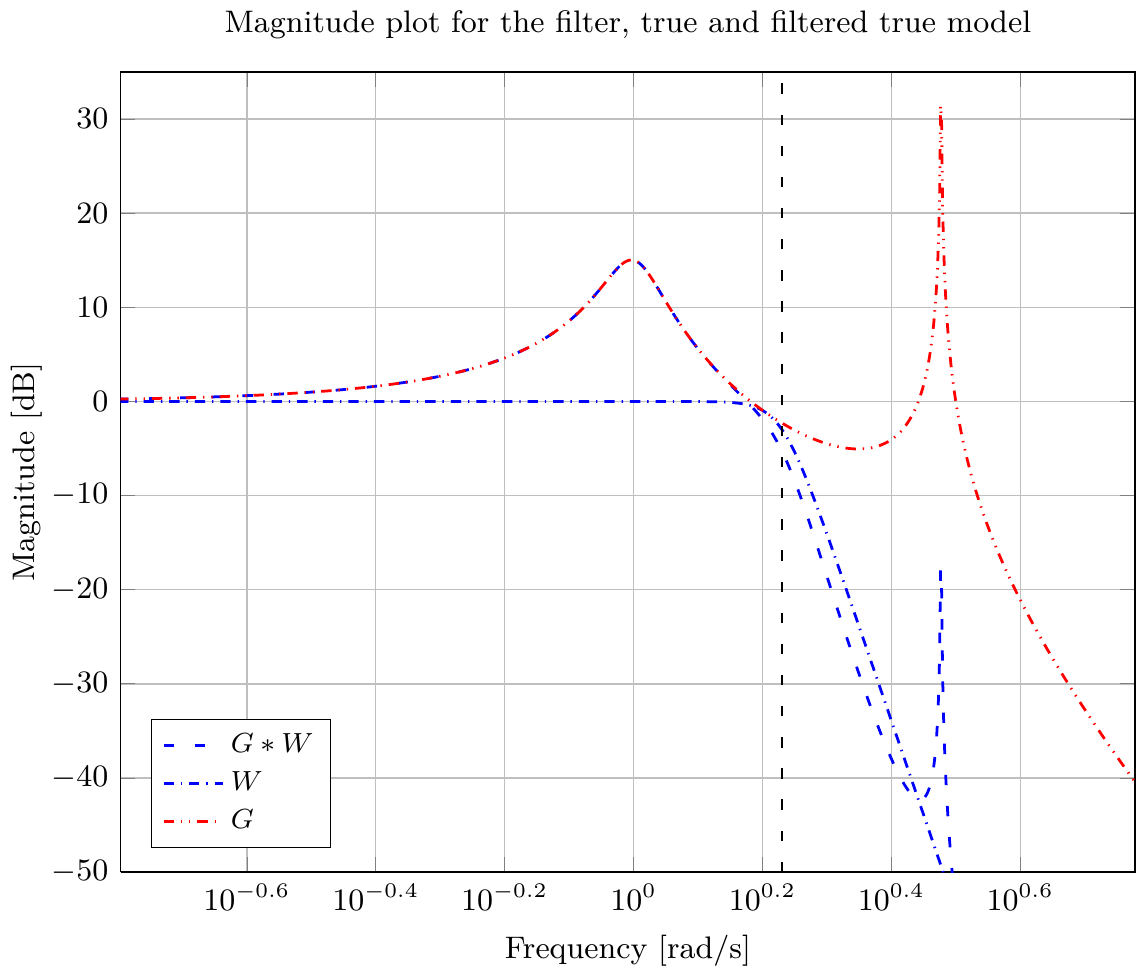}
  \caption{%
    Magnitude plot of the given model, the filtered model and the filter used in Example \ref{ex:twoModeSys}. \figComment}\label{fig:twoModeSysFilter}
\end{figure}
\begin{figure}
  \centering
  \includegraphics[width=\figwidth]{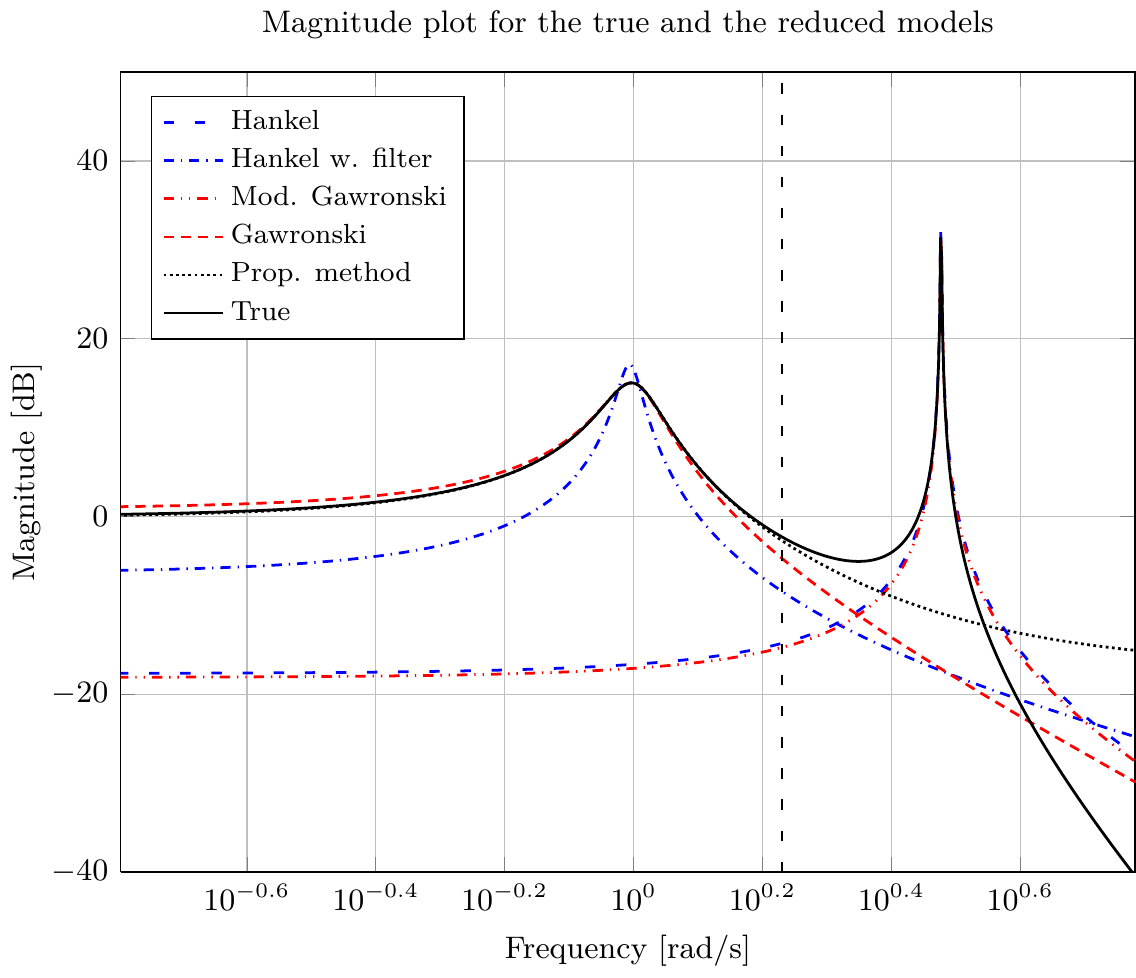}
  \caption{Magnitude plot of the given and reduced order models in Example \ref{ex:twoModeSys} for $\omega\in[0,1.7]$. The dashed black vertical line denotes $\omega=1.7$. We see in the figure that the proposed method and the Gawronski method finds models which are good approximations of the true model on the given frequency interval. \figComment}
  \label{fig:twoModeSysFull}
\end{figure}%
\begin{figure}
  \centering
  \includegraphics[width=\figwidth]{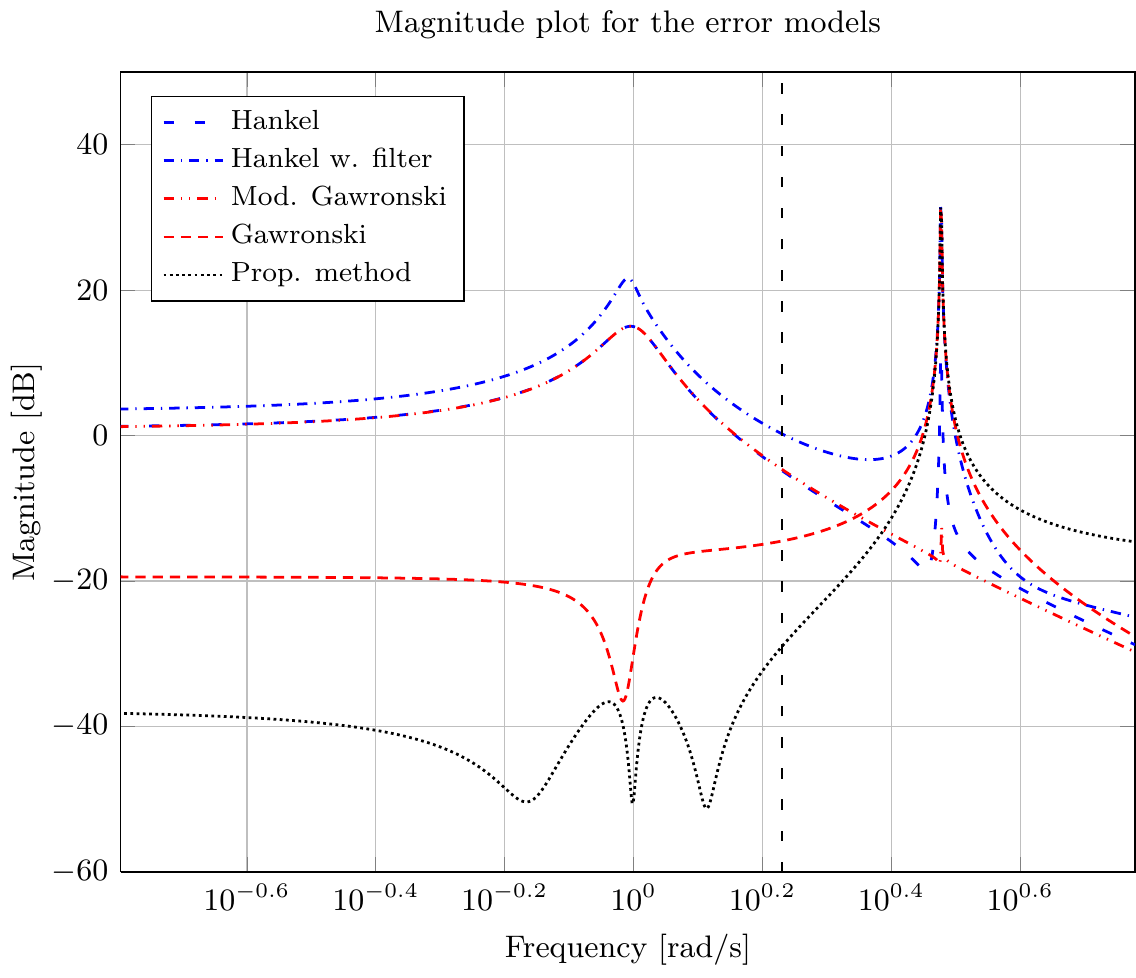}
  \caption{Magnitude plots of the error models in Example \ref{ex:twoModeSys} for $\omega\in[0,1.7]$. The dashed black vertical line denotes $\omega=1.7$. We see in the figure that the proposed method and the Gawronski method finds accurate approximations to the true model on the given frequency interval. \figComment}
  \label{fig:twoModeSysError}
\end{figure}
\begin{table}[!htb]
  \centering
  \caption{Numerical results for Example \ref{ex:twoModeSys}}
  \label{tab:twoModeSys}
  \begin{tabular}{r|c|c|c|c}
    & $\norm{G-\hat G}_{\Htwo,\omega}$ & $\frac{\norm{G-\hat G}_{\Htwo,\omega}}{\norm{G}_{\Htwo,\omega}}$ & $\frac{\norm{G-\hat G}_{\Hinf,\omega}}{\norm{G}_{\Hinf,\omega}}$ & $\real \lambda_\text{max}$ \\
    \hline
          Hankel &     1.77e+00 &     1.01e+00 &     1.00e+00 &  -1.59e-03 \\ 
Hankel w. filter &     2.93e+00 &     1.67e+00 &     2.14e+00 &  -4.03e-02 \\ 
Mod. Gawronski &     1.77e+00 &     1.01e+00 &     1.00e+00 &  -1.51e-03 \\ 
   Gawronski &     9.14e-02 &     5.21e-02 &     3.35e-02 &  -9.88e-02 \\ 
Prop. method &     8.51e-02 &     4.85e-02 &     3.26e-02 &  -9.94e-02 \\ 

  \end{tabular}
\end{table}
\end{example}

\begin{example}[Example 1 in \cite{GawronskiJ:1990}]\label{ex:gawronski}
  In this example we reuse Example 1 from \cite{GawronskiJ:1990}. The model, which is a spring-damper model with three masses, has six states and we will reduce the model to three states and limit the frequency interval to $\omega\in[1.5,3.2]$. To help the Hankel method we create a band pass Butterworth filter of order 10 with cut off frequencies of $1.5$ and $3.2$ rad/s, see Figure \ref{fig:gawronskiFilter}. The results from the different methods can be seen in Figures \ref{fig:gawronskiFull} and \ref{fig:gawronskiError} and Table \ref{tab:gawronski}. 
The proposed method and the Gawronski method are also in this example successful in finding low order models that approximate the given model on the given frequency range, and all the reduced order models are stable.
The proposed method and the Gawronski method captures the correct frequency interval with good accuracy. The Hankel method on the filtered model captures the correct resonance frequency, however, with the worst overall $\Htwo$-measure. The other two methods misses the relevant frequency interval. Only the proposed method, the Gawronski method and the Hankel method with the filtered model are able to find the right resonance peak, whilst the modified Gawronski and the Hankel method are completely off.
\begin{figure}
   \centering
  \includegraphics[width=\figwidth]{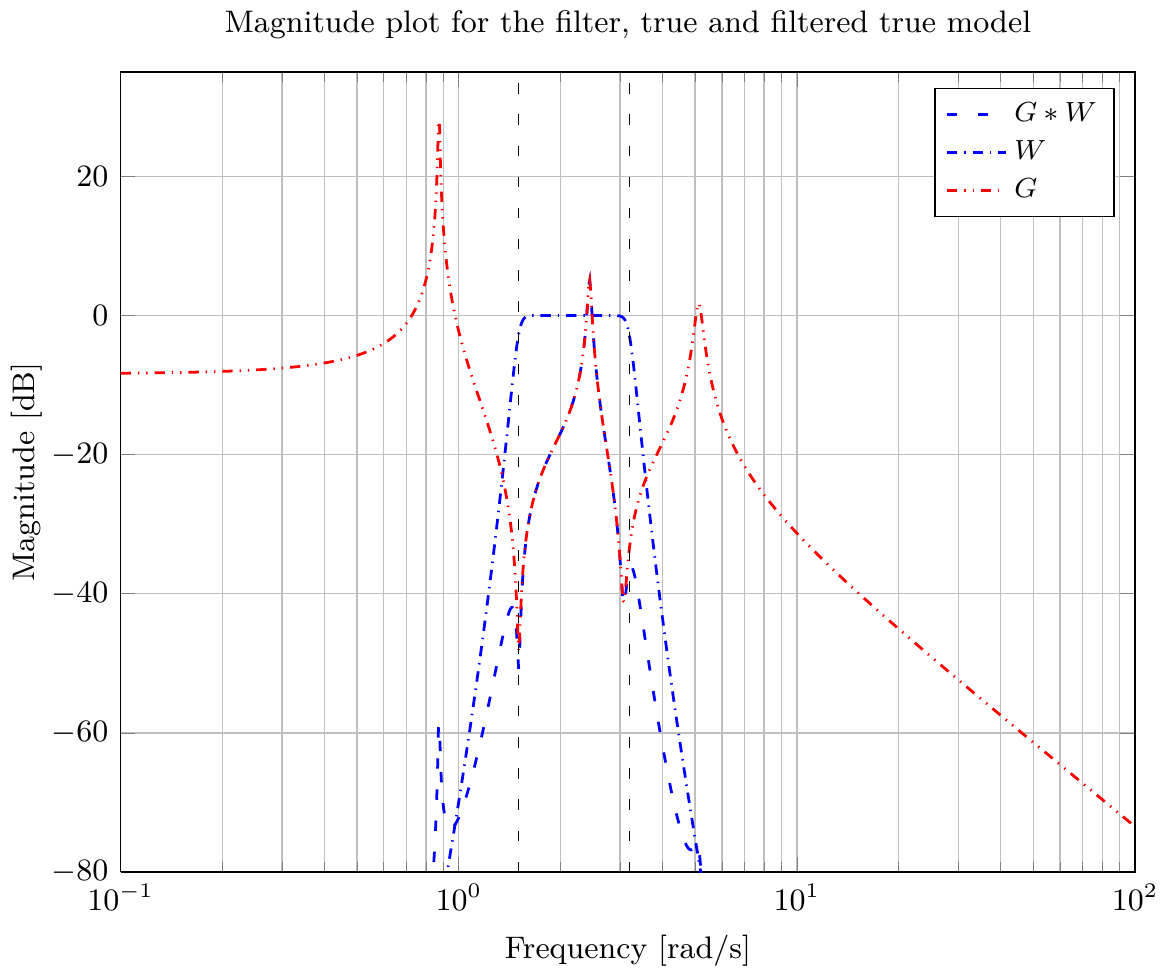}
  \caption{%
    Magnitude plot of the given model, the filtered model and the filter used in Example \ref{ex:gawronski}. \figComment}\label{fig:gawronskiFilter}
\end{figure}
\begin{figure}
  \centering
  \includegraphics[width=\figwidth]{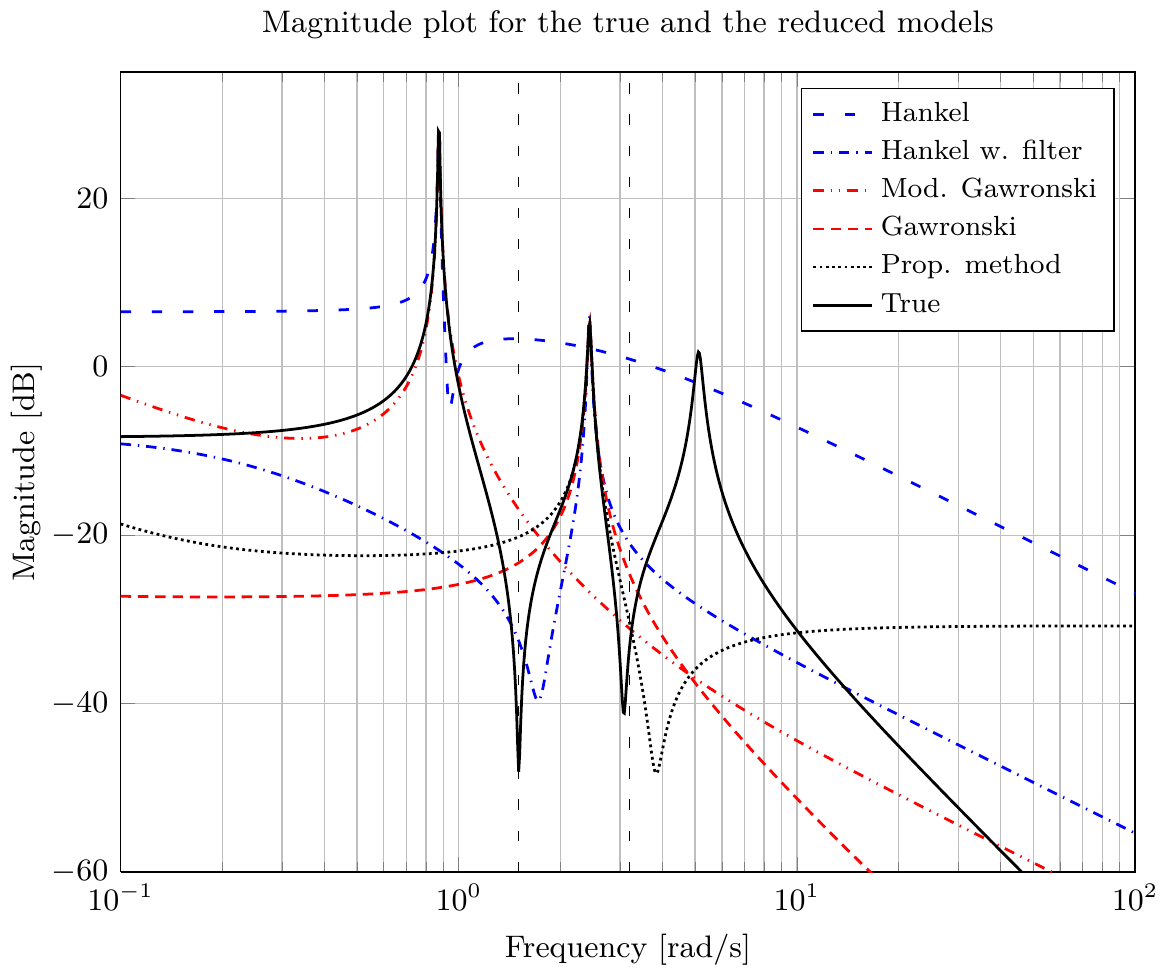}
  \caption{Magnitude plot of the given and reduced order models in Example \ref{ex:gawronski} for $\omega\in[1.5,3.2]$. The dashed black vertical lines denotes $\omega=1.5$ and $3.2$. We see that proposed method and the Gawronski method tries to capture the correct resonance frequency with good accuracy, also the Hankel method on the filtered model does this but not as good. The other methods misses the relevant frequency interval. \figComment}
  \label{fig:gawronskiFull}
\end{figure}%
\begin{figure}
  \centering
  \includegraphics[width=\figwidth]{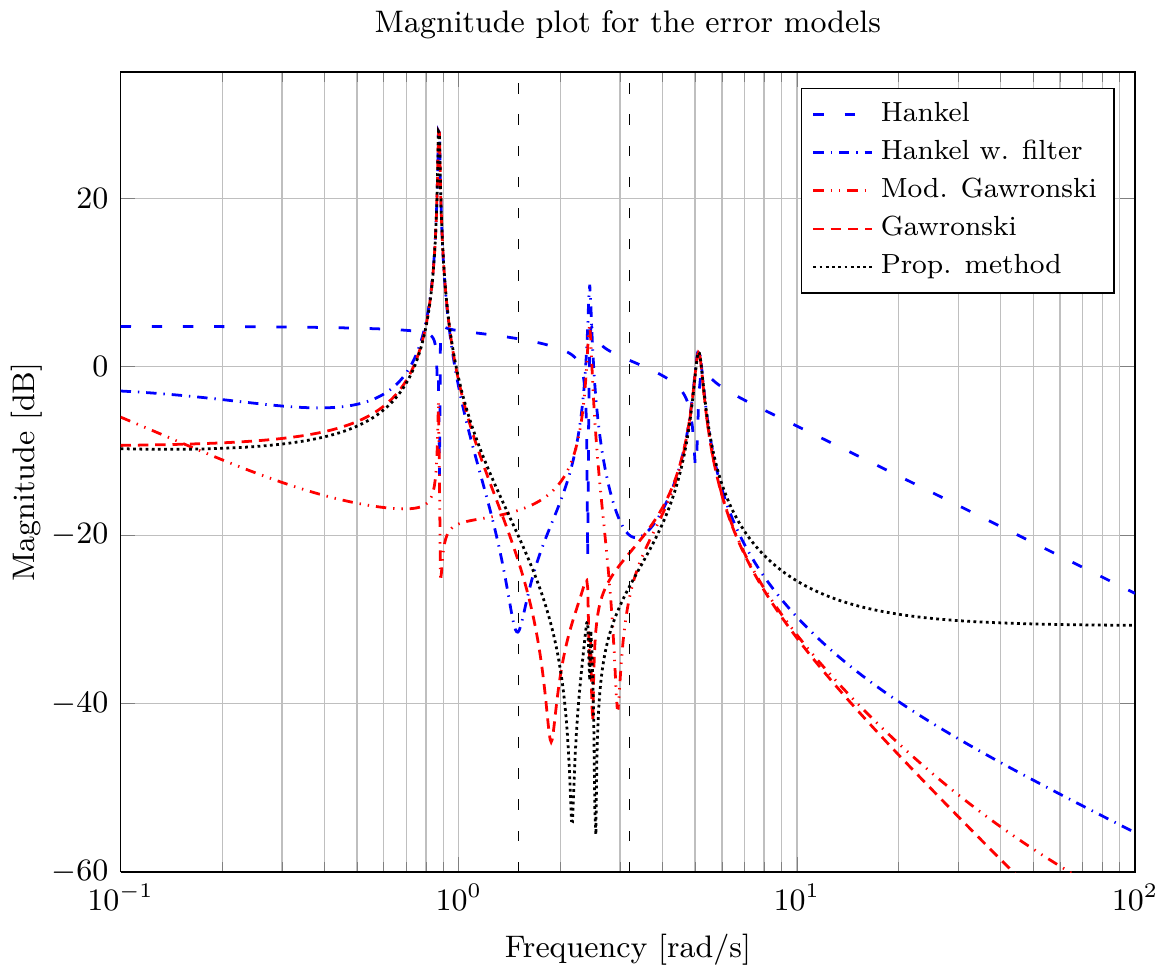}
  \caption{Magnitude plots of the error models in Example \ref{ex:gawronski} for $\omega\in[1.5,3.2]$. The dashed black vertical lines denotes $\omega=1.5$ and $3.2$. The proposed method and the Gawronski method finds the best approximations on the relevant frequency interval. \figComment}
  \label{fig:gawronskiError}
\end{figure}
  \begin{table}[!htb]
    \centering
    \caption{Numerical results for Example \ref{ex:gawronski}}
    \label{tab:gawronski}
    \begin{tabular}{r|c|c|c|c}
    & $\norm{G-\hat G}_{\Htwo,\omega}$ & $\frac{\norm{G-\hat G}_{\Htwo,\omega}}{\norm{G}_{\Htwo,\omega}}$ & $\frac{\norm{G-\hat G}_{\Hinf,\omega}}{\norm{G}_{\Hinf,\omega}}$ & $\real \lambda_\text{max}$ \\
    \hline
            Hankel &     9.27e-01 &     3.05e+00 &     9.59e-01 &  -3.54e-03 \\ 
Hankel w. filter &     4.58e-01 &     1.50e+00 &     1.71e+00 &  -2.02e-02 \\ 
Mod. Gawronski &     3.10e-01 &     1.02e+00 &     9.89e-01 &  -3.78e-03 \\ 
   Gawronski &     3.30e-02 &     1.09e-01 &     4.36e-02 &  -2.86e-02 \\ 
Prop. method &     2.72e-02 &     8.94e-02 &     5.34e-02 &  -2.86e-02 \\ 

    \end{tabular}
  \end{table}
\end{example}

\begin{example}[Aircraft example]\label{ex:cofcluo}
  The model in this example is a model with 22 states that describes the longitudinal motion of an aircraft, see \cite{VargaHP:2012}. We will reduce this model to 15 states and limit the frequency range to $\omega\in[0,15]$. To help the Hankel method we create a low pass Butterworth filter of order 10 with a cut off frequency of 15 rad/s, see Figure \ref{fig:cofcluoFilter}. The results from the different methods can be seen in Figures \ref{fig:cofcluoFull} and \ref{fig:cofcluoError} and Table \ref{tab:cofcluo}. In this example the Gawronski method results in a model which at a first glance looks very good, however the model is unstable. The other methods capture the given model with varying accuracy. Since the Gawronski method generates an unstable model in this example, the proposed method is initialized with the model from the Hankel reduction instead. The proposed method finds the most accurate model on the given frequency range.
\begin{figure}
   \centering
  \includegraphics[width=\figwidth]{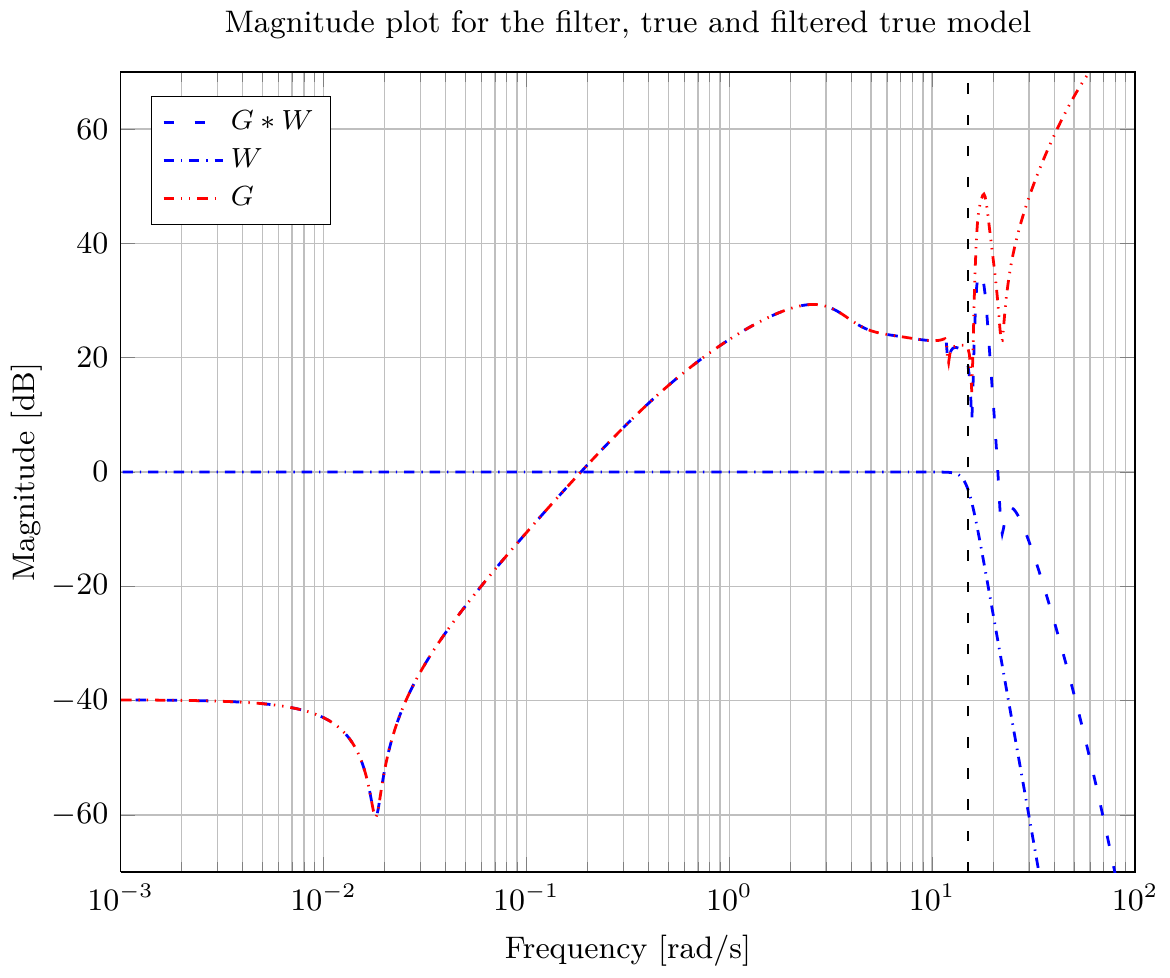}
  \caption{%
    Magnitude plot of the given model, the filtered model and the filter used in Example \ref{ex:cofcluo}. \figComment}\label{fig:cofcluoFilter}
\end{figure}
\begin{figure}
  \centering
  \includegraphics[width=\figwidth]{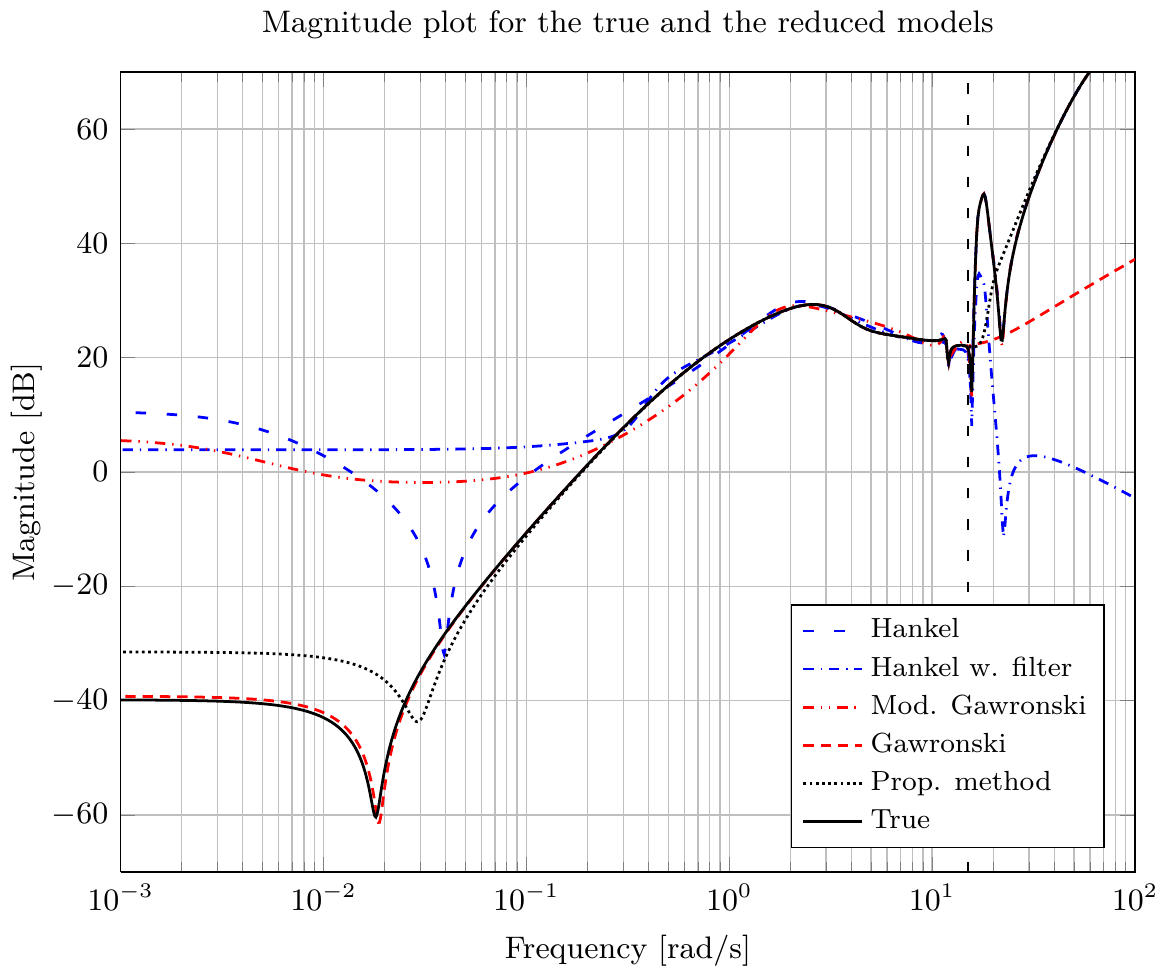}
  \caption{Magnitude plot of the given and reduced order models in Example \ref{ex:cofcluo} for $\omega\in[0,15]$. The dashed black vertical line denotes $\omega=15$. The Gawronski method looks to have found the best approximation, however the model found by the Gawronski method is unstable. The proposed method finds the best stable approximation of the true model. \figComment}
  \label{fig:cofcluoFull}
\end{figure}%
\begin{figure}
  \centering
  \includegraphics[width=\figwidth]{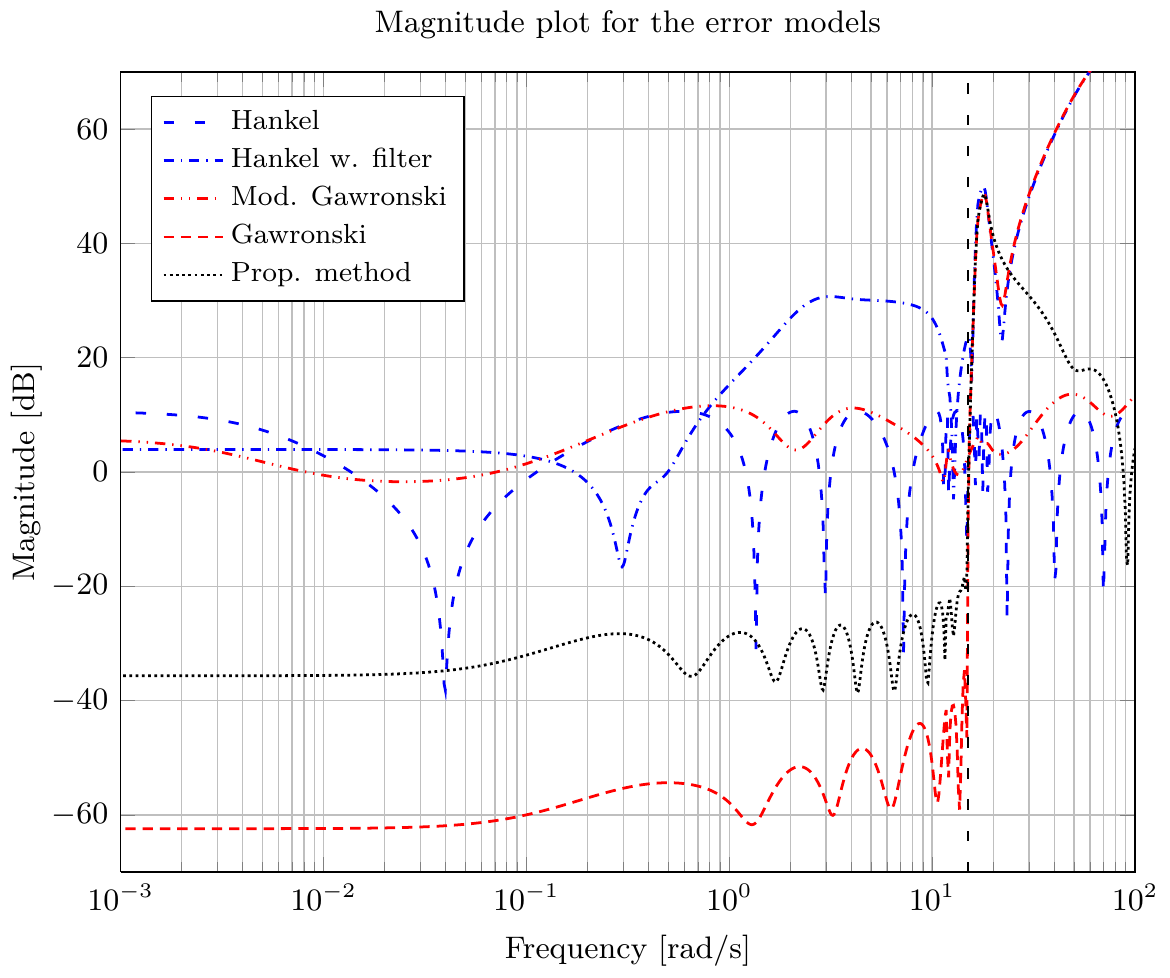}
  \caption{Magnitude plots of the error models in Example \ref{ex:cofcluo} for $\omega\in[0,15]$. The dashed black vertical line denotes $\omega=15$. The Gawronski method looks to have found the approximation with the smallest error, however the model found by the Gawronski method is unstable. The proposed method finds the model with the lowest error which is stable. \figComment}
  \label{fig:cofcluoError}
\end{figure}
  \begin{table}[!htb]
    \centering
    \caption{Numerical results for Example \ref{ex:cofcluo}}
    \label{tab:cofcluo}
    \begin{tabular}{r|c|c|c|c}
    & $\norm{G-\hat G}_{\Htwo,\omega}$ & $\frac{\norm{G-\hat G}_{\Htwo,\omega}}{\norm{G}_{\Htwo,\omega}}$ & $\frac{\norm{G-\hat G}_{\Hinf,\omega}}{\norm{G}_{\Hinf,\omega}}$ & $\real \lambda_\text{max}$ \\
    \hline
            Hankel &     5.15e+00 &     1.39e-01 &     1.16e-01 &  -4.85e-03 \\ 
Hankel w. filter &     5.10e+01 &     1.37e+00 &     1.17e+00 &  -1.24e-01 \\ 
Mod. Gawronski &     5.00e+00 &     1.35e-01 &     1.30e-01 &  -3.19e-03 \\ 
   Gawronski &          -- &          -- &     -- &   4.67e+00 \\ 
Prop. method &     1.68e-01 &     4.51e-03 &     1.51e-02 &  -1.77e-01 \\ 

    \end{tabular}
  \end{table}
\end{example}
In all three of the above examples we observe that the proposed method finds an at least as good model as the method proposed in \cite{GawronskiJ:1990}, but with the proposed method we can guarantee that the reduced order model is stable and also impose structure in the system matrices. In the first two examples it takes less than one second and in the third example about 20 seconds for the proposed method to find a reduced order model.

\section{Conclusions}

In this paper we have proposed a method, based on nonlinear optimization, that uses the frequency-limited Gramians introduced in \cite{GawronskiJ:1990}, and the method does not have the drawback of finding unstable models. We use these Gramians to construct a frequency-limited $\Htwo$-norm which describes the cost function of the optimization problem. We derive a gradient of the proposed cost function which enables us to use off-the-shelve optimization software to solve the problem efficiently. The derivation of the method also enables us to impose structural constraints, \eg, upper triangular $\m A$-matrix, in the system matrices. The derivation follows closely the technique in \cite{Petersson:2010} and it is easy to extend the method to identify \lpv-models. We also presented three examples of different sizes and characteristics to show the applicability of the method. The method uses a nonlinear optimization approach on a nonconvex problem, hence the unavoidable problem of local minima are present. However, we consider this as a two step approach where you use your favorite reduction method as an initial point to this method which will refine the solution, as we have shown in the examples.


\bibliography{IEEEabrv,modRedArticleSCL}
\bibliographystyle{plain}

\appendix

\section{Proof of Lemma \ref{lem:realPart}}
\label{sec:proof-lemma-}

We want to show that
\begin{equation}
  \label{eq:17}
  \m S_{\omega} = \frac{\ii}{2\pi}\ln\parens{(\m A + \ii\omega\ident)(\m A - \ii\omega\ident)^\inv} = 
    \real\brackets{\frac{\ii}{\pi}\ln\parens{-\m A-\ii\omega\ident}}
\end{equation}
assuming that $\m A$ is Hurwitz, \ie, $\m A$ is a real matrix and if $\lambda$ is an eigenvalue to $\m A$ then we have that $\real\parens{\lambda} < 0$. We can decompose $\m A$ as $\m A = \m V\m D\m V^\inv$ where $\m D$ is a diagonal matrix with the eigenvalues to $\m A$ on the diagonal and $\m V$ is the matrix with the eigenvectors to $\m A$ as its columns. For a real matrix we have that the eigenvectors for real eigenvalues are always real and that the eigenvectors for a complex-conjugate pair are also complex-conjugate.
We can write $\m A$ as
\begin{align}
  \label{eq:131}
  \m A = & \m V\m D\m V^\inv =
  \begin{bmatrix}
    v_1 v_2 \cdots v_n
  \end{bmatrix}
  \begin{bmatrix}
    \lambda_1 & 0 & \cdots &  & 0 \\
    0 & \lambda_2 & & &  \\
    \vdots &  & \ddots &  & \vdots \\
    0 & 0 & \cdots & & \lambda_n
  \end{bmatrix}
  \begin{bmatrix}
    \tilde v_1^\+ \\ \tilde v_2^\+ \\ \vdots \\ \tilde v_n^\+
  \end{bmatrix} \nonumber \\
  = & v_1\tilde v_1 \lambda_1 + v_2\tilde v_2\lambda_2 + \dots + v_n\tilde v_n\lambda_n
\end{align}
and $f(\m A)$ as
\begin{align}
  \label{eq:13}
  f(\m A) = & \m Vf(\m D)\m V^\inv =
  \begin{bmatrix}
    v_1 v_2 \cdots v_n
  \end{bmatrix}
  \begin{bmatrix}
    f(\lambda_1) & 0 & \cdots &  & 0 \\
    0 & f(\lambda_2) & & &  \\
    \vdots &  & \ddots &  & \vdots \\
    0 & 0 & \cdots & & f(\lambda_n)
  \end{bmatrix}
  \begin{bmatrix}
    \tilde v_1^\+ \\ \tilde v_2^\+ \\ \vdots \\ \tilde v_n^\+
  \end{bmatrix} \nonumber \\
  = & v_1\tilde v_1 f(\lambda_1) + v_2\tilde v_2f(\lambda_2) + \dots + v_n\tilde v_nf(\lambda_n).
\end{align}
For a real eigenvalue, $\lambda$, in $\m A$ with corresponding eigenvector, $v$, represented as one term in equation \eqref{eq:13} and
\begin{equation}
  \label{eq:19}
  f(\lambda) = \ii\ln\parens{\frac{\lambda+\ii\omega}{\lambda-\ii\omega}} = \ii\ln\parens{-\lambda-\ii\omega} - \ii\ln\parens{-\lambda+\ii\omega}
\end{equation}
we can write this term as
\begin{multline}
  \label{eq:1244}
  v \tilde v^\+ \parens{\ii\ln\parens{-\lambda-\ii\omega} -  \ii\ln\parens{-\lambda+\ii\omega}} = v \tilde  v^\+ \parens{\ii\ln\parens{-\lambda-\ii\omega} + \overline{ \ii\ln\parens{-\lambda-\ii\omega}}} \\
  = 2v\tilde v^\+\real\brackets{\ii\ln\parens{-\lambda-\ii\omega}} = 2\real\brackets{v\tilde v^\+\ii\ln\parens{-\lambda-\ii\omega}} = 2\real\brackets{v\tilde v^\+g(\lambda)}.
\end{multline}
If we now instead look at two terms in \eqref{eq:13} corresponding to a complex-conjugate pair of eigenvalues, $\lambda$ and $\overline\lambda$, with corresponding eigenvectors, $v$ and $\overline v$, then we can write this as
\begin{multline}
  \label{eq:1455}
  v\tilde v^\+f(\lambda) + \overline v \overline{\tilde v^\+} f(\overline\lambda) = v\tilde v^\+ \ii\ln\parens{-\lambda-\ii\omega} - v\tilde v^\+ \ii\ln\parens{-\lambda+\ii\omega} \\
  + \overline v \overline{\tilde v^\+}\ii\ln\parens{-\overline\lambda-\ii\omega} - \overline v\overline{\tilde v^\+} \ii\ln\parens{-\overline\lambda+\ii\omega} = v \tilde v^\+g(\lambda) + \overline{\overline v\overline{\tilde v^\+} g(\overline\lambda)} + \\
  + \overline v \overline{\tilde v^\+}g(\overline \lambda) + \overline{v\tilde v^\+ g(\lambda)} = 2\real\brackets{v\tilde v^\+ g(\lambda)} + 2\real\brackets{\overline v\overline{\tilde v^\+}g(\overline\lambda)}.
\end{multline}
This means that for a matrix, $\m A$, which is Hurwitz, we have
\begin{multline}
  \label{eq:1555}
  f(\m A) = \m V f(\m D)\m V^\inv =  v_1\tilde v_1 f(\lambda_1) + v_2\tilde v_2f(\lambda_2) + \dots + v_n\tilde v_nf(\lambda_n) \\
  = 2\real\brackets{v_1\tilde v_1 g(\lambda_1) + v_2\tilde v_2g(\lambda_2) + \dots + v_n\tilde v_ng(\lambda_n)} = 2\real g(\m A),
\end{multline}
\ie, we have
\begin{equation}
  \label{eq:18}
  \m S_{\omega} = \frac{\ii}{2\pi}\ln\parens{(\m A + \ii\omega\ident)(\m A - \ii\omega\ident)^\inv} = 
    \real\brackets{\frac{\ii}{\pi}\ln\parens{-\m A-\ii\omega\ident}}.
\end{equation}

\section{Derivation of the Gradient w.r.t. \texorpdfstring{$\mh A$}{Ahat}}
\label{sec:derGradA}

In this appendix we present the differentiation of the cost function with respect to $\mh A$.

The cost function is
\begin{align}
  \norm{E}^2_{\Htwo,\omega} & = \trace \left( \m B^\+\m Q_\omega\m B + 2\m B^\+\m Y_\omega\mh B + \mh B^\+\mh Q_\omega\mh B\right) \nonumber \\
    & + 2\trace\brackets{\m C\m S_\omega\m B + \m D\frac{\omega}{2\pi} -  \parens{\mh C\mh S_\omega\mh B + \mh D\frac{\omega}{2\pi}}}\parens{\m D^\+ - \mh D^\+}\label{eq:costfcnB2}
\end{align}
and by looking at the equations
\begin{subequations}\label{eq:14}
\begin{align}
  \label{eq:1312}
  \m A^\+\m Y_\omega + \m Y_\omega\mh A - \m S^\herm_\omega\m C^\+\mh C - \m C^\+\mh C\mh S_\omega & = \m 0, \\
  \mh A^\+\mh Q_\omega + \mh Q_\omega\mh A + \mh S^\herm_\omega\mh C^\+\mh C + \mh C^\+\mh C\mh S_\omega & = \m 0,
\end{align}
\end{subequations}
we observe that $\mh Q_\omega$ and $\m Y_\omega$ depend on $\mh A$ which we need to keep in mind when differentiating \eqref{eq:costfcnB2} with respect to $\mh A$. Hence, $\left[\frac{\partial\norm{E}^2_{\Htwo,\omega}}{\partial\mh A}\right]_{ij}$ becomes
\begin{equation}\label{eq:gradtemp2}
  \left[\frac{\partial\norm{E}^2_{\Htwo,\omega}}{\partial\mh A}\right]_{ij} = \trace
  \left(2\mh B\m B^\+\frac{\partial\m Y_\omega}{\partial\hat a_{ij}} + \mh B\mh
    B^\+\frac{\partial\mh Q_\omega}{\partial\hat a_{ij}} - 2 \mh C\frac{\partial\mh S_\omega}{\partial\hat a_{ij}}\mh B \parens{\m D^\+ - \mh D^\+}\right),
\end{equation}
where $\frac{\partial\m Y_\omega}{\partial\hat a_{ij}}$ and $\frac{\partial\mh Q_\omega}{\partial\hat a_{ij}}$ depend on $\mh A$ via the differentiated versions of the equations in \eqref{eq:14},
\begin{subequations}\label{eq:lyapperturbed2}
\begin{equation}
  \mh A^\+\frac{\partial\m Y^\+_\omega}{\partial\hat a_{ij}} + \frac{\partial\m
    Y^\+_\omega}{\partial\hat a_{ij}}\m A + \frac{\partial\mh A^\+}{\partial\hat a_{ij}}\m Y^\+_\omega - \frac{\partial \mh S_\omega^\herm}{\partial \hat a_{ij}}\mh C^\+\m C = \m 0 , \label{eq:lyapperturbed21}
\end{equation}
\begin{equation}
\mh A^\+\frac{\partial\mh Q_\omega}{\partial\hat a_{ij}} + \frac{\partial\mh Q_\omega}{\partial\hat a_{ij}}\mh A + \frac{\partial\mh A^\+}{\partial\hat a_{ij}}\mh Q_\omega
  + \mh Q_\omega\frac{\partial\mh A}{\partial\hat a_{ij}} \\
  + \frac{\partial \mh S_\omega^\herm}{\partial \hat a_{ij}}\mh C^\+\mh C + \mh C^\+\mh C \frac{\partial \mh S_\omega}{\partial \hat a_{ij}}= \m 0 .\label{eq:lyapperturbed22}
\end{equation}
\end{subequations}
To be able to substitute $\frac{\partial\mh Q_\omega}{\partial\hat a_{ij}}$ and $\frac{\partial\m Y_\omega}{\partial\hat a_{ij}}$ to something more computationally tractable we use the following lemma.
\begin{lemma}\label{lem:sylvester2}
  If $\m M$ and $\m N$ satisfy the Sylvester equations
  \begin{equation*}
    \m A\m M + \m M\m B + \m C = \m 0, \quad \m N\m A + \m B\m N + \m D = \m 0,
  \end{equation*}
  then $\trace \m C\m N = \trace \m D\m M$.
\end{lemma}

Studying Lemma \ref{lem:sylvester2}, the two factors in front of $\frac{\partial\m Y_\omega}{\partial\hat a_{ij}}$ and $\frac{\partial\mh Q_\omega}{\partial\hat a_{ij}}$ in \eqref{eq:gradtemp2} and the structure of the Lyapunov/Sylvester equations in \eqref{eq:lyapperturbed2}, $\mh A^\+ \cdot + \cdot \mh A + \star = 0$ and $\mh A^\+ \cdot + \cdot \m A + \star = 0$, brings us to the conclusion that, to do the substitution, we need to solve two additional Lyapunov/Sylvester equations, namely
\begin{subequations}\label{eq:lyapOrig2}
\begin{align}
  \m A\m X + \m X\mh A^T + \m B\mh B^T & = \m 0, \label{eq:lyapsylOrig21}\\
  \mh A\mh P + \mh P\mh A^T + \mh B\mh B^T & = \m 0. \label{eq:lyapsylOrig22}
\end{align}
\end{subequations}
Note that $\mh P$ in \eqref{eq:lyapsylOrig22} is the controllability Gramian for the reduced order model.

Rewriting \eqref{eq:gradtemp2} using Lemma \ref{lem:sylvester2}, \eqref{eq:lyapperturbed2} and \eqref{eq:lyapOrig2} leads to
\begin{align}
  \label{eq:82}
  \left[\frac{\partial\norm{E}^2_{\Htwo,\omega}}{\partial\mh A}\right]_{ij} = & 2\trace\brackets{{\frac{\partial\mh A^\+}{\partial\hat a_{ij}} \parens{\m Y_\omega^\+\m X + \mh Q_\omega\mh P}} 
    + {\frac{\partial\mh S^\herm_\omega}{\partial\hat a_{ij}}\parens{\mh C^\+\mh C\mh P - \mh C^\+\m C\m X}}} \nonumber \\
   & - 2\trace\brackets{\frac{\partial\mh S_\omega}{\partial\hat a_{ij}}\brackets{\mh B\parens{\m D^\+ - \mh D^\+}\mh C}}.
\end{align}
What remains is to rewrite the two last terms in \eqref{eq:82}, which includes $\frac{\partial\mh S_\omega}{\partial\hat a_{ij}}$ and $\frac{\partial\mh S^\herm_\omega}{\partial\hat a_{ij}}$.
Recall the definition of $\mh S_\omega$,
\begin{equation}
  \label{eq:6}
   \mh S_\omega = \real\brackets{\frac{\ii}{\pi}\ln\parens{-\mh A - \ii\omega\ident}} = \real\brackets{\frac{\ii}{\pi}\ln\parens{r(\mh A)}}
\end{equation}
and differentiate with respect to an element in $\mh A$, \ie,  $a_{ij}$. This yields
\begin{equation}
  \label{eq:72}
  \frac{\partial \mh S_\omega}{\partial a_{ij}} = \real\brackets{\frac{\ii}{2\pi}L\parens{r(\mh A),\frac{\partial r(\mh A)}{\partial a_{ij}}}}
  = \real\brackets{\frac{\ii}{2\pi}L\parens{r(\mh A),-\frac{\partial \mh A}{\partial a_{ij}}}}
\end{equation}
where $L(\m A,\m E)$ is the Frech\'et derivative of the matrix logarithm, see \cite{Higham:2008}, with
\begin{align}
  \label{eq:62}
  L(\m A,\m E) = & \int_0^1 \parens{t(\m A-\ident)+\ident}^\inv\m E\parens{t(\m A-\ident)+\ident}^\inv \der t, \\
  r(\mh A) = & -\mh A - \ii\omega\ident.
\end{align}
The function $L(\m A,\m E)$ can be efficiently evaluated using the algorithm by \cite{Higham:2008}.

By substituting \eqref{eq:72} into \eqref{eq:82} and using \eqref{eq:62} with the fact that we can interchange the $\trace$-operator and the integral we obtain
\begin{multline}
  \label{eq:93}
  \left[\frac{\partial\norm{E}^2_{\Htwo,\omega}}{\partial\mh A}\right]_{ij} =  2\trace\brackets{{\frac{\partial\mh A^\+}{\partial\hat a_{ij}} \parens{\m Y_\omega^\+\m X + \mh Q_\omega\mh P}}
  +  {\frac{\partial\mh S^\herm_\omega}{\partial\hat
      a_{ij}}\parens{\mh C^\+\mh C\mh P - \mh C^\+\m C\m X}}} \\
   - 2\trace\brackets{\frac{\partial\mh S_\omega}{\partial\hat a_{ij}}\brackets{\mh B\parens{\m D^\+ - \mh D^\+}\mh C}} \\
   = 2\trace\parens{\frac{\partial\mh A^\+}{\partial\hat a_{ij}} \parens{\m Y_\omega^\+\m X + \mh Q_\omega\mh P}}
   - 2\trace\bigg(\frac{\partial\mh A^\+}{\partial\hat a_{ij}} \real\left[\frac{\ii}{\pi}L\parens{r(\mh A),\m V}\right]^\+\bigg)  \\
   =  \trace\parens{\frac{\partial\mh A^\+}{\partial\hat a_{ij}} \brackets{2\parens{\m Y_\omega^\+\m X + \mh Q_\omega\mh P} -2\m W }},
\end{multline}
where
\begin{align}
  \label{eq:111}
  \m W = & \real\left[\frac{\ii}{\pi} L\parens{r(\mh A),\m V}\right]^\+, \\
  \m V = & \mh C^\+\mh C\mh P - \mh C^\+\m C\m X-\mh C^\+\parens{\m
      D - \mh D}\mh B^\+ .
\end{align}

\end{document}